\documentclass[11pt,a4paper]{article}
\usepackage{jheppub}
\usepackage[T1]{fontenc}

\usepackage{mathtools}
\usepackage{tikz}
\usepackage{tikz-cd}
\usepackage[titletoc]{appendix}

\newcounter{counter2}
\numberwithin{equation}{section}
\numberwithin{counter2}{section}

\newtheorem{thm}[counter2]{Theorem}

\newtheorem{lemma}[counter2]{Lemma}

\newcommand{\koniec}{\begin{flushright}  $\Box $ \end{flushright}}

\DeclareSymbolFont{script}{U}{eus}{m}{n}
\DeclareMathSymbol{\Wedge}{0}{script}{"5E}

\newcounter{mnotecount}[section]

\renewcommand{\themnotecount}{\thesection.\arabic{mnotecount}}

\newcommand{\mnote}[1]
{\protect{\stepcounter{mnotecount}}$^{\mbox{\footnotesize
$
\bullet$\themnotecount}}$ \marginpar{
\raggedright\tiny\em
$\!\!\!\!\!\!\,\bullet$\themnotecount: #1} }

\newcommand{\CP}{\mathbb{CP}}                 

\newcommand{\C}{\mathbb{C}}

\newcommand{\R}{\mathbb{R}}

\def\p{\partial}
\def\be{\begin{equation}}

\def\ee{\end{equation}}

\def\bea{\begin{eqnarray}}
\def\eea{\end{eqnarray}}

\title{\boldmath Heavenly equations in de Sitter space}

\author[a]{Maciej Dunajski}
\author[a]{and Timothy Moy}

\affiliation[a]{Department of Applied Mathematics and Theoretical Physics, University of Cambridge,\\ Wilberforce Road, Cambridge CB3 0WA, UK.}

\emailAdd{m.dunajski@damtp.cam.ac.uk}
\emailAdd{tjahm2@cam.ac.uk}

\toccontinuoustrue

\abstract{
We demonstrate that all anti-self-dual Einstein metrics with non--zero cosmological constant 
$\Lambda$ locally arise from solutions of a single second order PDE introduced by Lipstein and Nagy.  We show how this  equation fits into the heavenly formalism of Pleba\'nski, and establish a Lax pair.  Finally we show how Pleba\'nski's second heavenly equation arises in the limit as $\Lambda\rightarrow 0$.
}

\begin{document}
\maketitle
\flushbottom

\begin{center}
{\em In memory of Jerzy Lukierski (1936–2026)}
\end{center}

\section{Introduction}
In the seminal work \cite{Pleb} of Pleba\'nski,  two normal forms for four--dimensional Ricci-flat anti-self-dual  (ASD) complex spacetimes were presented.  The first form is adapted to a local choice of compatible integrable complex structure.  
The second form,  obtained by a  Darboux transformation from the first,  is adapted to a single foliation by self-dual surfaces.  The corresponding equations expressing the ASD Ricci-flat condition are known as the first and second heavenly equations,  respectively. The twistor approach to these equations was developed in \cite{DM02}, and
in \cite{DM24} it was shown how the second equation arises
as an infinitesimal,  but still non-linear,  limit of the first.  
\par 
Przanowski \cite{PrzanKill,PrzanHerm} obtained a normal form and equation for ASD vacuum metrics with non-zero cosmological constant $\Lambda$.  The equation is analogous to the first heavenly equation and the corresponding 
twistor theory was  developed  in  \cite{APV,Hoeg}. 
Finley and Pleba\'nski \cite{FinPleb},  and later Przanowski and Chudecki \cite{CP} have considered normal forms,  and the corresponding equations,  for such spacetimes adapted to a single integrable self-dual foliation.  
More recently,  Lipstein and Nagy \cite{LN} provided an elegant metric ansatz for which the ASD-Einstein equation with $\Lambda=-3$  reduces to a single second order PDE.  Restoring  the explicit dependence of $\Lambda$ the Lipstein--Nagy (LN) 
equation for $W=W(x, y, w, z)$ is
\be\label{wxyz_heavenly}
\Psi\equiv (\p_x\p_w+\p_y\p_z)(\phi^{-1} W)-\phi\{\p_x(\phi^{-1}W), \p_y(\phi^{-1}W)\}_*=0,
\quad \mbox{where}\quad \phi=w-\Lambda x
\ee
and  $\{F, G\}_*=\{F, G\}_{xy}+2\Lambda\phi^{-1}\langle F, G\rangle_y$ with 
\[
\{F, G\}_{xy}\equiv\p_x F \p_y G-\p_x G \p_y F, \quad \langle F, G\rangle_y\equiv F\p_y G-G\p_y F.
\]
The authors of \cite{LN}, perhaps unaware of \cite{FinPleb, CP}  have not demonstrated that
their equation describes  all ASD $\Lambda$--vacua.  Our note aims to fill in this gap. In \S\ref{secproof} we shall prove
\begin{thm}
\label{mainthm}
Let $(X, g)$ be a complexified Einstein manifold with scalar curvature $24\Lambda$ and ASD Weyl tensor. There exist local coordinates  $(x, y, w, z)$ and a function $W$ on $X$ such that
\begin{align}\label{LN_metric}
g = \frac{2}{\phi^2}\left[ dw\,dx + dy\,dz + W_{yy} dw^2 - 2(W_{xy} + \frac{2\Lambda}{\phi}W_y ) 
dw\,dz +  ( W_x + \frac{4\Lambda}{\phi} W )_x \ dz^2  \right]
\end{align}
and $W$ satisfies (\ref{wxyz_heavenly}).  Conversely,  $W$ satisfying (\ref{wxyz_heavenly}) yields ASD $\Lambda$--Einstein metric (\ref{LN_metric}).  
\end{thm}
Our proof will demonstrate that the equation of Chudecki and Przanowski in \cite{CP} and the LN equation arise from a different gauge choice applied to the master dispersionless integrable system \cite{DFK} describing all ASD conformal structures.

In \S\ref{seclax} we shall find a Lax representation $[L_0, L_1]=0$ for (\ref{wxyz_heavenly}) given by
\begin{eqnarray}
\label{laxpairintro}
L_0 &=& \partial_{w}+(W_{x} + \frac{4\Lambda}{\phi} W)_y\partial_y - W_{yy}\partial_x - \lambda \partial_y + \beta_0 \partial_{\lambda}\\ 
L_1 &=& \partial_{z} -(W_x+\frac{4\Lambda}{\phi}W)_x \partial_y+W_{xy}\partial_{x} + \lambda \partial_{x} + \beta_1 \partial_{\lambda} \nonumber
\end{eqnarray}
with $\beta_0=-\Box W_y, \beta_1=\Box (W_x+4\Lambda\phi^{-1}W)-3\Lambda\Psi$
and
\be
\label{boxoperator}
\Box=\p_w\p_x+\p_z\p_y-W_{yy}\p_x^2-(W_x+4\Lambda\phi^{-1}W)_x\p_y^2+ 2(W_x+2\Lambda\phi^{-1}W)_y\p_x\p_y.
\ee
In \S\ref{seclimit} we show how taking $\Lambda\rightarrow 0$ combined with a coordinate transformation reduces (\ref{wxyz_heavenly}) to the second heavenly equation. In \S\ref{secexamples} we give some examples of solutions to (\ref{wxyz_heavenly}).
\section*{Acknowledgements}
We are grateful to Adam Kmec for pointing out equation (\ref{wxyz_heavenly}) to us,
and to Adam Chudecki and Maciej Przanowski for explaining their work \cite{CP}. 
We thank the anonymous referees for
their detailed comments which have lead to improvements in the manuscript.
Our research was supported by the Simons Foundation
grant  SFI-MPS-T-Institutes-00010825,  and by the State Treasury funds as part of a
task commissioned by the Minister of Science and Higher Education under the project Organization of the Simons Semesters at the Banach Center - New Energies in 2026-2028 (MNiSW/2025/DAP/491). T.M.  is supported by Cambridge Australia Scholarships.  
\section{The master equations of anti-self-duality}
We will use the Einstein-summation convention throughout,  using spinor indices taking values $A \in \{0,1\}$ and $\epsilon^{AB},  \epsilon_{AB}$ the standard symplectic matrices satisfying $\epsilon_{AB}\epsilon^{BC} = \delta_{A}{}^{B}$.  We will raise and lower indices using these objects.  So,  for example 
$W^{A} := \epsilon^{AB}W_{B}$.  

An adapted local Pleba\'nski coordinate system  $z^A=(w, z)$ and $\theta_A=(x, y)$
for the general ASD conformal structure $(X, [g])$ was obtained in \cite{DFK}.  In such coordinates
\begin{eqnarray}\label{general_integrable}
g &=& \frac{2}{\phi^2}(d \theta_{A} dz^{A} + \frac{\partial W_{A}}{\partial \theta^{B}}dz^{A} dz^{B}),\\
&=& \frac{2}{\phi^2}(e^{00'} \odot e^{11'} - e^{10'} \odot e^{01'}),\nonumber
\end{eqnarray}
for some holomorphic functions $W^{A}(z, \theta)$,   $\phi(z,\theta)$.  Here $e^{AA'}$ is the basis
of $\Lambda^{1}(X)$  satisfying $e^{BB'}(E_{AA'}) = \delta_{B}{}^{A}\delta_{B'}{}^{A'}$ and
\be
\label{Evectors}
E_{A0'} = \frac{\partial}{\partial \theta^{A}},\quad
E_{A1'} = \frac{\partial}{\partial z^{A}} + \epsilon^{BC}\frac{\partial W_{B}}{\partial \theta^{A}}\frac{\partial}{\partial \theta^{C}}.
\ee
We have 
 \begin{align}
 [E_{A0'},E_{B1'}] =  -\frac{\partial^2 W^{C}}{\partial \theta^{A} \partial \theta^{B}}E_{C0'}, \quad 
 [E_{A1'},E_{B1'}] = -2\Box W_{[A} E_{B]0'}\label{brac2} ,  
 \end{align}
where the second-order differential operator $\Box :=  E_{11'}E_{00'}-E_{01'}E_{10'}$.  
The full ASD condition requires imposing two additional PDE:
\begin{align}
E_{A0'}(\square W^{A}) = 0, \quad 
E_{A1'}(\square W^{A}) = 0 \label{ASD2}.
\end{align}
There is considerable coordinate freedom on offer here.  Firstly,  we may choose to adapt the coordinates to any self-dual foliation\footnote{Anti-self-duality is equivalent to the existence of a one-parameter family of such foliations.},  which then takes the explicit form
$
{\mathcal D} = \operatorname{span}\big\{ \frac{\partial}{\partial \theta^{0}}, \frac{\partial}{\partial \theta^{1} }\big\}.  
$
Then,  for local coordinates $\tilde{z}^{A}$ on the leaf space $X / {\mathcal D}$,  we may take as coordinates on $X$ 
\begin{align}\label{fibre_lift}
\tilde{z}^{A}(z),  \quad \tilde{\theta}^{A} = \Gamma(z)\frac{\partial \tilde{z}^{A}}{\partial z^{B}}\theta^{B} +\gamma^A(z),  
\end{align}
for arbitrary holomorphic functions $\gamma^{A}(z),  \Gamma(z)$.  The metric takes the form (\ref{general_integrable}) and the form of the ASD equations is preserved after 
redefining $\phi$ and $W^{A}$ appropriately.  In particular,  $\phi$ transforms as 
$
\tilde{\phi} = \Gamma^{1/2}(z)\Delta^{1/2}\phi$ where $\Delta := \det \frac{\partial \tilde{z}^{A}}{\partial z^{B}} .  
$
\section{A normal form for an ASD Einstein metric}
\label{secproof0}
\begin{lemma}\label{Lipstein_gauge_fix_prop}
Given an ASD Einstein metric with $\Lambda := R/24\neq 0$,  we may choose Pleba\'nski coordinates $(z^{A},  \theta^{A})$ so that there exists a function $W=W(z, \theta)$ such that 
\begin{align}\label{W_potential}
\phi = {\Lambda} J_{A}\theta^{A} + K_{A}z^{A}, \quad  W_{A} = \frac{\partial W}{\partial \theta^{A}} - \frac{4\Lambda}{\phi}J_{A}W.
\end{align}
where constants $J_{A},  K_{A}$ ($A=0, 1$) that satisfy $J^{A}K_{A} = 1$ may be chosen arbitrarily.   
\end{lemma}
\noindent
{\bf Proof.}
The proof is based on computations of the Ricci tensor \cite{FinPleb, Timthesis}.
The $e^{A0'} \odot e^{B0'}$ components of the trace-free Ricci tensor vanish if and only if 
\begin{align}
\label{phiformula}
\phi = {\Lambda} S_{A}(z)\theta^{A} + T(z)
\end{align}
for some functions $S_{A}(z),  T(z)$,  where the factor of $\Lambda$ has been inserted for convenience and is consistent with the fact that the scalar curvature vanishes if $\phi$ is a function of $z^A$.  

We rescale the $\theta^{A}$ coordinates by a function of $z^{A}$,  which rescales $\phi$ so that the $1$-form
$S_{A}(z)dz^{A}$
is closed.  This implies there exist functions $\tilde{z}^{A}(z)$ that solve
$
J_{A}\frac{\partial \tilde{z}^{A}}{\partial z^{B}} = S_{B}(z)
$
for some constants $J_{A}$ (not all zero) that we may choose.  We then define $\tilde{\theta}^{A}$ by (\ref{fibre_lift}) with $\gamma^{A}(z) = 0$ so 
$
\tilde{\phi} = \Gamma^{1/2}\Delta^{1/2}(\Gamma^{-1}\Lambda J_{A}\tilde{\theta}^{A} + T(\tilde{z})).  
$
Setting $\Gamma = \Delta$ gives 
$
\tilde{\phi} = {\Lambda}{J}_{A}\tilde{\theta}^{A} + \tilde{T}(\tilde{z})
$ for $\tilde{T}(\tilde{z}) = \Delta T(z)$.  
Lastly we add functions $\gamma^{A}(z)$ to $\tilde{\theta}^{A}$ to put $\phi$ in the form (\ref{W_potential}).  \par 
 We continue imposing the Einstein equations.  The vanishing of the $e^{A0'} \odot e^{B1'}$ components of the trace-free Ricci tensor yield linear PDE
\begin{align}\label{integrability_1}
\frac{\partial^2 V}{\partial \theta^{A} \partial \theta^{B}} = 0,  \quad
\mbox{where}\quad
V := \epsilon^{AB} \phi \frac{\partial W_{A}}{\partial \theta^{B}} - 4{\Lambda} J^{A}W_{A}
= \epsilon^{AB}\phi^{5}\frac{\partial}{\partial \theta^{B}}(\phi^{-4}W_{A}).  
\end{align}
We have freedom to redefine the functions $W_A$ by
$
\tilde{W}_{A} = W_{A} + \chi(z)\epsilon_{AB}\theta^{B} + \zeta_A(z),  
$
for arbitrary $\chi(z),\zeta_{A}(z)$.  In terms of $V$,  since it satisfies (\ref{integrability_1}),  this means we may take 
\begin{align}\label{integrability condition}
V = 4\gamma(z)K_{A}\theta^{A} + \alpha(z) J_{A}\theta^{A} + \beta (z)
\end{align}
for arbitrary $\alpha(z),  \beta(z)$,  while $\gamma(z)$ is fixed.  
If we take $\alpha = \beta = 0$,   
then the integrability condition (\ref{integrability_1}) implies the existence of a function $W$ such that
\begin{align}
W_{A} = \frac{\partial W}{\partial \theta^{A}} - \frac{4\Lambda}{\phi}J_{A}W - \frac{\gamma(z)}{\Lambda}K_{A}K_{B}\theta^{B}.  
\end{align}
Computing the scalar curvature
$
R = 24\Lambda(1-\gamma(z))
$
gives $\gamma(z) = 0$.  \koniec
\section{The Lipstein-Nagy equation}
\label{secproof}
Having obtained a normal form for the metric depending on only one function $W$,  we will now show that the equation of Lipstein and Nagy is necessary and sufficient (after an appropriate redefinition of $W$) for the metric to satisfy the equations (\ref{ASD2}) and the parts of the Einstein condition which we have yet to impose.  
\vskip5pt
\noindent
{\bf Proof of Theorem \ref{mainthm}.}
The first step is to calculate:  
\begin{align}\label{Box_calc}
\square W_{A} = \phi \frac{\partial \Psi}{\partial \theta^{A}} - 3\Lambda J_{A}\Psi + K^{B}\frac{\partial^2 \tilde{W}}{\partial \theta^{B} \partial \theta^{A}} - \Lambda J^{B}\frac{\partial^2 \tilde{W}}{\partial z^{B} \partial \theta^{A}} - 2 \Lambda^2 J^{B}J^{C}\frac{\partial \tilde{W}}{\partial \theta^{B}}\frac{\partial^2 \tilde{W}}{\partial \theta^{A} \partial \theta^{C}}
\end{align}
where
\begin{align}
\Psi := \epsilon^{AB}\frac{\partial^2 \tilde{W}}{\partial \theta^{A} \partial z^{B}} - \frac{\phi}{2}\epsilon^{AB}\epsilon^{CD} \frac{\partial^2 \tilde{W}}{\partial \theta^{A} \partial \theta^{C}}\frac{\partial^2 \tilde{W}}{\partial \theta^{B} \partial \theta^{D}} {-} 2 \Lambda J^{C}\epsilon^{AB}\frac{\partial \tilde{W}}{\partial \theta^{A}}\frac{\partial^2 \tilde{W}}{\partial \theta^{C}\partial \theta^{B}}, \quad \tilde{W}\equiv \phi^{-1}W.
 \end{align}
All but one term of $\square W_{A}$ is a total derivative with respect to the $\theta^{A}$ coordinates.  Therefore the first equation in (\ref{ASD2}) simplifies to the third-order PDE
\begin{align}\label{I}
J^{A}\frac{\partial \Psi}{\partial \theta^{A}} = 0.  
\end{align}
Given the normal form as in Lemma \ref{Lipstein_gauge_fix_prop}, the non-zero components of the Einstein equations are the $e^{A1'} \odot e^{B1'}$ components,  which are equivalent to
\begin{align}\label{nonlinear}
\frac{\partial^2 (\phi^{-1}\Psi)}{\partial \theta^{A} \partial\theta^{B}} = 0.  
\end{align}
Using (\ref{Box_calc}) as well as (\ref{I}) and (\ref{nonlinear})  we expand the remaining ASD equation (\ref{ASD2}) as:
\begin{align}\label{II}
\phi \epsilon^{AB}
\frac{\partial^2\Psi}
     {\partial z^A\partial\theta^B}
-2K^A\frac{\partial\Psi}{\partial\theta^A}
-2\Lambda J^A\frac{\partial\Psi}{\partial z^A}=0
\end{align}
So $\Psi = 0$ implies anti-self-duality and the Einstein condition.  

 The converse is as follows:
The function $W$ satisfying (\ref{W_potential}) is only defined up to
$
\hat{W} = W + \rho(z)\phi^{4}.  
$
Under such a redefinition
\begin{align}\label{freedom}
\hat{\Psi} = \Psi - {3\Lambda}\phi^{2}J^{A}\frac{\partial \rho(z)}{\partial z^{A}} - 6 \Lambda \phi\rho(z).  
\end{align}
The general solution to the system (\ref{I},  \ref{II},  \ref{nonlinear}) is 
\begin{align}
\Psi = 3\Lambda\phi^{2}J^{A}\frac{\partial F(z)}{\partial z^{A}} + 6\Lambda \phi F(z)
\end{align}
for arbitrary $F(z)$.  So we may,  by redefining $\Psi$ as per (\ref{freedom}) set $\Psi = 0$.  The end result of the analysis is  that the metric (\ref{general_integrable}) with $(\phi, W_A)$ as in Lemma \ref{Lipstein_gauge_fix_prop}  such that
\begin{align}\label{LN equation}
\epsilon^{AB}\frac{\partial^2 \tilde{W}}{\partial \theta^{A} \partial z^{B}} - \frac{\phi}{2}\epsilon^{AB}\epsilon^{CD} \frac{\partial^2 \tilde{W}}{\partial \theta^{A} \partial \theta^{C}}\frac{\partial^2 \tilde{W}}{\partial \theta^{B} \partial \theta^{D}} {-} 2 \Lambda J^{C}\epsilon^{AB}\frac{\partial \tilde{W}}{\partial \theta^{A}}\frac{\partial^2 \tilde{W}}{\partial \theta^{C}\partial \theta^{B}} = 0 
\end{align}
is ASD and Einstein.  Conversely,  for any ASD Einstein metric,  there exist  coordinates $(z^{A},\theta^{A})$ such that $\phi$ and $W_{A}$ are as above,  and $\Psi = 0$ holds.  
Setting $z^A=(w, z), \theta^A=(y, -x), J_{A}=(0, 1),  K_{A}=(1, 0)$ reduces
$\Psi=0$ to the LN equation (\ref{wxyz_heavenly}) and
yields the statement of
Theorem \ref{mainthm}. The $\square$ operator in (\ref{brac2}) is given by (\ref{boxoperator}).\koniec
\subsection{Lax pair}
\label{seclax}
Consider a pair of vector fields on $X\times \CP^1$ indexed by $A=0, 1$:
\begin{align}
L_{A} = E_{A1'} - \lambda E_{A0'} + \beta_{A}\frac{\partial}{\partial \lambda}, \quad\mbox{where}\quad
\beta_{A} = -\square W_{A} - 3\Lambda J_{A}\Psi.  
\end{align}
Using (\ref{brac2}) we see that the $E_{A0'}$ components of $[L_{0}, L_{1}]$ vanish if and only if $\beta_{A} = -\square W_{A}$,  
which holds if and only if $\Psi = 0$.  The remaining components vanish with (\ref{ASD2}),  which we have 
shown follow from $\Psi = 0$.  We have thus exhibited a Lax pair for the equation (\ref{wxyz_heavenly}).
This could be a starting point of the inverse--scattering construction analogous to \cite{MS}.
\subsection{The Ricci-flat limit}
\label{seclimit}
If we set $\Lambda = 0$,  then the metric (\ref{LN_metric}) becomes
\begin{align}
g=\frac{2}{w^2}(dwdx+dzdy+W_{xx}dz^2+W_{yy}dw^2-2W_{xy} dwdz).
\end{align}
Letting
\begin{align}
\tilde{w}=-\frac{1}{2w^2}, \quad\tilde{z}=z, \quad \tilde{x}=wx,\quad \tilde{y}=\frac{y}{w^2},\quad \Theta=W-\frac{1}{2}\frac{xy^2}{w}
\end{align}
yields the second heavenly form \cite{Pleb}
\begin{align}
g=2(d\tilde{w}d\tilde{x}+d\tilde{z}d\tilde{y}+ \Theta_{\tilde{x}\tilde{x}} d\tilde{z}^2+\Theta_{\tilde{y}\tilde{y}}  d\tilde{w}^2-2\Theta_{\tilde{x}\tilde{y}} d\tilde{w}d\tilde{z})
\end{align}
where $\Theta$ satisfies Pleba\'nski's second heavenly equation 
\begin{align}
(\p_{\tilde{x}}\p_{\tilde{w}}+\p_{\tilde{y}}\p_{\tilde{z}})(\Theta)-\{\p_{\tilde{x}}\Theta, \p_{\tilde{y}}\Theta\}_{\tilde{x}\tilde{y}}=0.
\end{align}
\par 
It is an open problem to find an analogous limiting procedure between the  Przanowski equation
\cite{PrzanKill} and the first heavenly equation.
\subsection{Examples}
\label{secexamples}
A solution to (\ref{wxyz_heavenly})  given by an arbitrary function $W = \phi f(x, z)$
yields the ASD cosmological plane wave in \cite{DT}:
\begin{align}
g = \frac{2}{\phi^2}[ dw\,dx + dy\,dz + ( \phi f_{xx} + 2\Lambda f_x) dz^2 ].
\end{align}
Another class of solutions can be found by demanding that the linear and non--linear parts
of (\ref{wxyz_heavenly}) vanish separately.  The solution is analogous to the Sparling-Tod metric \cite{ST}:
\begin{align}
W = \frac{\phi \cdot F(\frac{w}{wx+zy},\frac{z}{wx+zy})}{wx+zy}.  
\end{align}
\subsection{Reality conditions}
The normal form (\ref{LN_metric}) in Theorem \ref{mainthm} gives a complex metric. If $W$ is taken to be a real function of real coordinates $(x, y, w,  z)$ then the resulting metric has signature $(2, 2)$. While all Riemannian ASD Einstein metrics can also in principle be recovered from (\ref{LN_metric}), imposing the
reality conditions is less natural, as the coordinates are  adapted to a choice of a self--dual surface and all
such surfaces in Riemannian signature are necessarily complex. Here the original formulation of Przanowski
\cite{PrzanKill} is advantageous. The trivial solution $W=0$ together with the reality conditions
$y=\bar{z}, (x, w)\in \R^2, \Lambda>0$ is the de Sitter space.

\subsection{The Chudecki--Przanowski equation}
In our derivation of (\ref{wxyz_heavenly}) we used the coordinate freedom to set $\phi=w-\Lambda x$.
The alternative choice made in \cite{CP} exploits the coordinate freedom together with
(\ref{phiformula}) to make $\phi$ independent of $z^A$, and (using our conventions) set $\phi=x$. 
The PDE resulting from the ASD and Einstein conditions is then 
\be
x^{-1}(U_{xw}+U_{yz})-\{U_x, U_y\}_{xy}+
2x^{-1}\langle U_x, U_y\rangle_{y}+{\Lambda}{(6x)}^{-1}
U_{xx}=0.
\ee
This is the Chudecki--Przanowski equation ((2.27) in {\cite{CP}})  for a potential $U(x, y, w, z)$.
\section{Outlook: Celestial symmetries with $\Lambda\neq 0$}
The resurgence of interest in heavenly equations and their integrability comes from
the appearance of their infinite--dimensional symmetry algebras $Lw_{1+\infty}$ in the context of celestial holography \cite{AMS}. In the ASD Ricci--flat case these symmetries can be 
understood as coming from a recursion
operator on space--time, and its Penrose transform acting on cohomology classes in the twistor space
\cite{DM02} (see also \cite{Miller}). The presence of non--zero cosmological constant breaks down the SDiff$(\C^2)$ symmetries
of the heavenly equations. This can already be seen in the LN equation (\ref{wxyz_heavenly}) as
the underlying bracket deforms the Poisson bracket by a Wronskian Lie algebra of
Diff$(\C)$.  The celestial symmetries in the presence of $\Lambda$, and their twistor counterparts  have been studied in \cite{TZ} and \cite{BBHKMS}. 

 Our Lax pair (\ref{laxpairintro}) leads to a notion of a recursion operator: Let
$\psi(x, y, w, z, \lambda)=\sum_{n=0}^{\infty} \lambda^{-n}\phi_n$ be a twistor function holomorphic around $\lambda=\infty$ where
$\phi_n, n=0, 1, \dots$ are functions on $X$. The twistor condition $L_A(\psi)=0$ leads to a recursion operator ${\mathcal R}(\phi_n)=\phi_{n+1}$ defined in terms of the vector fields (\ref{Evectors})
\be
\label{recursionrel}
E_{A0'}(\phi_{n+1})=E_{A1'}(\phi_n)-(n-1)\beta_A(\phi_{n-1}).
\ee
We can start the recursion  letting $\phi_0$ be the coordinate functions $w$ or $z$ which gives
\[
 w\rightarrow y\rightarrow W_x+4\Lambda\phi^{-1}W\rightarrow\cdots\\, \qquad z\rightarrow -x\rightarrow W_y\rightarrow\cdots
\]
with the higher terms determined by the linear PDEs (\ref{recursionrel}).
In the case $\Lambda=0$ the recursion operator acts on solutions to the background--coupled wave 
equation \cite{DM02}.
In the presence of $\Lambda$ the functions $\phi_n$ are twisted by sections of powers of a non--trivial line--bundle (see \cite{Hoeg} for a description in the context of the Przanowski equation). It would  be interesting to see 
if this recursion procedure can be used to reconstruct the 
$\Lambda$--celestial symmetries and the associated hierarchy in the spirit of \cite{bogdanov}.

\end{document}